\documentclass[12pt]{article}
\usepackage{amscd,amsmath,amssymb,amstext,amsthm,exscale,latexsym}
\usepackage{graphicx}
\newcommand {\al}   {\alpha}       \newcommand {\bt}  {\beta}
\newcommand {\g }   {\gamma}       
\newcommand {\dl}   {\delta}       \newcommand {\e }  {\epsilon}
        
\newcommand {\ve}   {\varepsilon}

\newcommand {\s }   {\sigma}       
         
\newcommand {\vf }  {\varphi}

\newcommand {\pl}   {\partial}     \newcommand {\nb}  {\nabla}
\renewcommand {\det}{{\sf\,det\,}}

       \renewcommand {\lim}{{\sf\,lim\,}}
\newcommand   {\ex}{{\sf\,e}}

     \newcommand   {\diag}{{\sf\,diag\,}}

\newcommand {\MM}  {{\mathbb M}}   
   
   \newcommand {\MR}  {{\mathbb R}}
   \newcommand {\MT}  {{\mathbb T}}
\newcommand {\MU}  {{\mathbb U}}   
   
   \newcommand {\MZ}  {{\mathbb Z}}

\newcommand {\CC }  {{\cal C}}

\newcommand {\Sa}  {{\textsc{a}}}   \newcommand {\Sb}  {{\textsc{b}}}

\newcommand {\Sg}  {{\textsc{g}}}   
   
   \newcommand {\Sl}  {{\textsc{l}}}
   \newcommand {\Sn}  {{\textsc{n}}}
   
   \newcommand {\Sr}  {{\textsc{r}}}

\newtheorem{lemma}{Lemma}[section]

\newtheorem{prop}{Proposition}[section]
\newtheorem{theorem}{Theorem}[section]

\theoremstyle{definition}
\newtheorem*{cor}{Corollary}
\newtheorem*{com}{Comment}
\begin{document}
\title     {On the existence of the global conformal gauge in string theory}
\author    {M. O. Katanaev
            \thanks{E-mail: katanaev@mi-ras.ru}\\ \\
            \sl Steklov Mathematical Institute,\\
           \sl ul. Gubkina, 8, Moscow, 119991, Russia}

\maketitle
\begin{abstract}
The global conformal gauge is playing the crucial role in string theory
providing the basis for quantization. Its existence for two-dimensional
Lorentzian metric is known locally for a long time. We prove that if a
Lorentzian metric is given on a plain then the conformal gauge exists globally
on the whole $\MR^2$. Moreover, we prove the existence of the conformal gauge
globally on the whole worldsheets represented by infinite strips with straight
boundaries for open and closed bosonic strings. The global existence of the
conformal gauge on the whole plane is also proved for the positive definite
Riemannian metric.
\end{abstract}
%******************************************************************************
\section{Introduction}
%*******************************************************************************
The (super)string theory attracts much interest in physics and mathematics for
the last fifty years (see, e.g., \cite{GrScWi87,BriHen88,BarNes90}). It is
usually considered as the basis for construction of the unified quantum theory
of all fundamental interactions including gravity. The crucial role in the
theory is played by the conformal gauge for a metric of Lorentzian signature
in which it is conformally flat. In fact, almost all results in string theory
are obtained using the assumption that the conformal gauge exists on the whole
string worldsheets represented by infinite strips with straight boundaries. For
example, the covariant and light cone quantizations use Fourier series which
exist only if the conformal gauge is applied on the whole string worldsheet.

The local existence of the conformal gauge is well known for a long time
(see, for $\CC^2$-metric, e.g., \cite{Petrov61}, Ch.~I, \S~6.1, or
\cite{Vladim71}, Ch.~I, \S~3.1). This gauge and boundary value problems
were also considered in \cite{VlaVol86R}.  The local existence of the gauge
is proved by writing down equations
for transformation functions and considering their integrability conditions
which guarantee the existence of solution in some neighbourhood of an
arbitrary point. However it is not enough. In string theory, it is assumed that
the conformal gauge exists on an infinite strip with straight boundaries. There
are two subtle questions: does the conformal gauge exist on the whole strip? and
can the boundaries be made straight? In the present paper, we answer these
questions affirmatively. The transition from local to global considerations is
based on the global existence theorem for the solution of the Cauchy problem for
two-dimensional hyperbolic differential equations with varying coefficients
(see, e.g., \cite{Hadama32}, book IV, ch.\ I).
This theorem is highly nontrivial, but allows one to make global statements.

The existence of the global conformal gauge, adopted in string theory, results
in beautiful and consistent theory. Therefore the theory deserves attention by
itself. But the question remains: are there other solutions of the initial
Nambu--Goto string which are not captured by the usual approach? We prove that
there are no such solutions. The main results of the present paper for
Lorentzian signature metric are published in \cite{Katana21B} without proofs.

The local existence of the conformal gauge (isothermal coordinates) for positive
definite Riemannian metric is also known in mathematics for a long time. The
proof for analytic metric is given e.g.\ in \cite{Petrov61}, Ch.~I, \S~6.4 and
\cite{Vladim71}, Ch.~I, \S~3.4 and for $\CC^3$-metric e.g.\ in \cite{Wolf72},
Theorem 2.5.14. In the present paper, we extend the proof to the whole Euclidean
plane.

In the next section, we introduce notation and write down equations of motion
with boundary conditions. Afterwards we consider infinite, open, and closed
strings in subsequent sections, respectively. Finally, we analyze the Riemannian
two-dimensional metric.

All functions are supposed to be sufficiently smooth by default.
%******************************************************************************
\section{The bosonic string}
%*******************************************************************************
Consider two manifolds: a plane $\MR^2$ with arbitrary global
coordinates $x=(x^\al):=(x^0,x^1):=(\tau,\s)$, $\al=0,1$, and $D$-dimensional
Minkowskian space $\MR^{1,D-1}$ with Cartesian coordinates $X=(X^\Sa)$,
$\Sa=0,1,\dotsc,D-1$, $D\ge2$, and the Lorentz metric
$\eta_{\Sa\Sb}:=\diag(+-\dotsc-)$. Let there be a smooth embedding
\begin{equation}                                                  \label{ubnchs}
  X:\qquad\MR^2\supset\overline\MU\ni\qquad(\tau,\s)\mapsto
  \big(X^\Sa(\tau,\s)\big)\qquad\in\MR^{1,D-1},
\end{equation}
of some closed subset $\overline\MU$ of a plane. We assume that $\MU$ is
connected and simply connected open subset in $\MR^2$. The embedding defines the
string worldsheet $\MM:=X(\overline\MU)$.

Embedding (\ref{ubnchs}) defines symmetric quadratic form with components
\begin{equation}                                                  \label{enbsty}
  h_{\al\bt}:=\pl_\al X^\Sa\pl_\bt X^\Sb\eta_{\Sa\Sb}
  =\pl_\al X^\Sa\pl_\bt X_\Sa.
\end{equation}
In general, this form may be negative definite, degenerate, or indefinite. We
assume that the embedding is such that
\begin{equation}                                                  \label{edsfwe}
\begin{split}
  (\pl_0 X)^2:=&\dot X^2:=\dot X^\Sa\dot X^\Sb\eta_{\Sa\Sb}>0,
\\
  (\pl_1 X)^2:=&X^{\prime2}:=X^{\prime\Sa}X^{\prime\Sb}\eta_{\Sa\Sb}<0,
\end{split}
\end{equation}
where the dot and prime denote differentiation with respect to $\tau$ and $\s$,
respectively, on $\MU$.
The vectors $\dot X$ and $X'$ are linearly independent on $\MM$. Here and in
what follows indices $\Sa,\Sb,\dotsc$ are often omitted. So global coordinates
$\tau,\s$ on $\MU$ are timelike and spacelike, respectively. Then the
determinant of the induced quadratic form is negative
\begin{equation}                                                  \label{eqjfyh}
  h:=\det h_{\al\bt}=\dot X^2X^{\prime2}-(\dot X,X')^2<0,
\end{equation}
where brackets denote the usual scalar product in $\MR^{1,D-1}$. Now the
embedding (\ref{ubnchs}) defines the Lorentzian metric on the string worldsheet
interior $\MU$ with signature $(+-)$.

Open string is the embedding (\ref{ubnchs}) of the closed straight strip
\begin{equation}                                                  \label{bdvgtr}
  -\infty<\tau<\infty,\qquad 0\le\s\le\pi
\end{equation}
with properties (\ref{edsfwe}). This strip is vertical if $\tau$ and $\s$
coordinate axes are depicted by vertical and horizontal straight lines on a
plain $\MR^2$, respectively.

Closed string is the embedding (\ref{ubnchs}) of the closed straight verical
strip
\begin{equation}                                                  \label{qndbfu}
  -\infty<\tau<\infty,\qquad -\pi\le\s\le\pi
\end{equation}
with identified boundaries. There are many ways to identify smoothly the
boundaries (\ref{qndbfu}). In string theory, we, first, impose the conformal
gauge on the metric on the same strip (\ref{qndbfu}) and, second, impose the
smooth periodicity conditions
\begin{equation}                                                  \label{annfgh}
  \pl^k_1 X^\Sa\big|_{\s=-\pi}=\pl^k_1 X^\Sa\big|_{\s=\pi},\qquad\forall\Sa,
  \forall\tau,\quad k=0,1,2,\dotsc,
\end{equation}
up to the needed order. It is the prime aim of the present paper to prove that
the conformal gauge on the same strips does exist.

Sure, a cylinder is not a simply connected manifold and cannot be covered by a
single coordinate chart. The domain (\ref{qndbfu}) is the fundamental domain for
a closed string worldsheet with identified boundaries.

A coordinate system defined on the domains for open (\ref{bdvgtr}) and for
closed (\ref{qndbfu}) strings we call {\em global coordinate system} on the
string worldsheets.

If infinite strips in the $\tau,\s$ plane have curved boundaries, then all of
them are diffeo\-morphic to strips (\ref{bdvgtr}) or (\ref{qndbfu}). Thus we did
not loose generality by specifying the coordinate range in the $\tau,\s$ plane.

The dynamics of the Nambu--Goto string is governed by the action which is
proportional to the string worldsheet area
\begin{equation}                                                  \label{ubxvgy}
  S_{\Sn\Sg}:=-\int_{\overline\MU}\!\!dx\sqrt{|h|}
  =-\int_{\overline\MU}\!\!d\tau d\s\sqrt{\displaystyle(\dot X,X')^2
  -\dot X^2X^{\prime2}},
\end{equation}
where $h:=\det h_{\al\bt}$. This action is invariant with respect to arbitrary
coordinate changes and global Lorentz transformations. It implies the
Euler--Lagrange equations
\begin{equation}                                                  \label{uvbsju}
  \frac1{\sqrt{|h|}}\frac{\dl S_{\Sn\Sg}}{\dl X_\Sa}=
  \square_{(h)} X^\Sa=h^{\al\bt}\nb_\al\nb_\bt X^\Sa
  =\frac1{\sqrt{|h|}}\pl_\al\left(\sqrt{|h|}h^{\al\bt}\pl_\bt X^\Sa\right)=0,
\end{equation}
where the two-dimensional wave operator $\square_{(h)}$ is build
by the induced metric $h_{\al\bt}$ (\ref{enbsty}) and $\nb_\al$ is the
covariant derivative with respective Christoffel's symbols.

We assume that ends of an open string are free, and then the action
(\ref{ubxvgy}) implies also the boundary conditions
\begin{equation}                                                  \label{ubvbxg}
  s^\bt\pl_\bt X^\Sa\big|_{\s=0,\pi}=0,
\end{equation}
where $s^\al$ are components of the spacelike vector which is perpendicular
to the boundaries with respect to the induced metric.

The action (\ref{ubxvgy}) does not yield any boundary condition for a closed
string. Instead, we have periodicity conditions (\ref{annfgh}) imposed by hands.

In string theory, the crucial role is played by the possibility to impose
global conformal gauge
\begin{equation}                                                  \label{edbfht}
  h_{\al\bt}=\ex^{2\phi}\eta_{\al\bt},\qquad\eta_{\al\bt}:=\diag(+-),
\end{equation}
where $\phi(x)$ is some sufficiently smooth function, on the whole string
worldsheet. The aim of the present paper is to prove that this conformal gauge
can be imposed on the same strips (\ref{bdvgtr}) and (\ref{qndbfu}) both for
open and closed strings with the same straight boundaries.
%******************************************************************************
\section{The idea of the proof}
%*******************************************************************************
The idea of the proof is the following. We construct two orthogonal vector
fields: the timelike $t=t^\al\pl_\al$ and spacelike
$s=s^\al\pl_\al$ vector fields such that the following conditions hold on the
whole string worldsheet $\MU\hookrightarrow\MR^{1,D-1}$:
\begin{equation}                                                  \label{axncjh}
  (t,s)=0,\qquad t^2+s^2=0,\qquad t^2>0,\qquad\forall x\in\MU,
\end{equation}
where the scalar product is defined by the induced metric $h_{\al\bt}$
(\ref{enbsty}):
\begin{equation*}
  (t,s):=t^\al s^\bt h_{\al\bt},\qquad t^2:=(t,t),\qquad s^2:=(s,s).
\end{equation*}
Then we find conditions for commutativity of these vector fields: $[t,s]=0$.
The next step is to find two families of integral curves
$x^\al(\tilde\tau,\tilde\s)$ which are defined by the system of
differential equations
\begin{equation}                                                  \label{ancjdy}
  \frac{\pl x^\al}{\pl\tilde\tau}=t^\al,\qquad
  \frac{\pl x^\al}{\pl\tilde\s}=s^\al,
\end{equation}
where $\tilde\tau$ and $\tilde\s$ are parameters along integral curves of vector
fields $t$ and $s$, respectively. The integrability conditions for this system
are fulfilled on the whole $\MU$:
\begin{equation*}
  \frac{\pl^2 x^\al}{\pl\tilde\tau\pl\tilde\s}
  -\frac{\pl^2 x^\al}{\pl\tilde\s\pl\tilde\tau}=\frac{\pl s^\al}{\pl\tilde\tau}
  -\frac{\pl t^\al}{\pl\tilde\s}=t^\bt\pl_\bt s^\al-s^\bt\pl_\bt t^\al
  =[t,s]^\al=0,
\end{equation*}
due to commutativity of vector fields. Consequently, there is a nondegenerate
coordinate transformation $(\tau,\s)\mapsto(\tilde\tau,\tilde\s)$ on the whole
worldsheet $\MU$.

In the new coordinate system, the induced metric $\tilde h_{\al\bt}$ is
conformally flat due to the properties of the vector fields (\ref{axncjh}):
\begin{equation}                                                  \label{ehsdhg}
\begin{split}
  \tilde h_{00}=&h_{\al\bt}\frac{\pl x^\al}{\pl\tilde\tau}
  \frac{\pl x^\bt}{\pl\tilde\tau}=t^2,
\\
  \tilde h_{01}=&h_{\al\bt}\frac{\pl x^\al}{\pl\tilde\tau}
  \frac{\pl x^\bt}{\pl\tilde\s}=(t,s)=0,
\\
  \tilde h_{11}=&h_{\al\bt}\frac{\pl x^\al}{\pl\tilde\s}
  \frac{\pl x^\bt}{\pl\tilde\s}=s^2=-t^2.
\end{split}
\end{equation}
The final step is the analysis of domains of the definition of parameters
$\tilde\tau$ and $\tilde\s$ of integral curves (\ref{ancjdy}) which are new
coordinates.

The last two conditions (\ref{axncjh}) imply the inequality $s^2<0$, i.e.\
the vector field $s$ is necessarily spacelike.

Here and in what follows, $\MU$ denotes either the whole Euclidean plane
(infinite string) or an open set (strip) on the plane
$(\tau,\s)\in\MU\subset\MR^2$ (open or closed string), where the induced metric
is nondegenerate. The boundaries $\pl\MU$ of open string, on which the metric
is degenerate, are considered separately.
%******************************************************************************
\section{Infinite string}
%*******************************************************************************
Let us start the detailed analysis. First, we consider the embedding
(\ref{ubnchs}) where $\MU=\MR^2$, i.e.\ the embedding of the whole plane
(infinite string). Arbitrary timelike and spacelike tangent vectors $T$ and $S$
to the string worldsheet in the embedding space $\MR^{1,D-1}$ can be decomposed
on tangent vectors $\dot X$ and $X'$:
\begin{equation}                                                  \label{abnxhh}
\begin{split}
  T=&A(\cosh\vf\dot X+\sinh\vf X'),
\\
  S=&B(\sinh\psi\dot X+\cosh\psi X'),
\end{split}
\end{equation}
where $A(x)\ne0$, $B(x)\ne0$ and $\vf(x),\psi(x)\in\MR$ are some functions.
Suppose that vectors $T$ and $S$ for $\vf=\psi=0$ are directed in the same way
as vectors $\dot X$ and $X'$, respectively. Then $A>0$ and $B>0$.
\begin{lemma}                                                     \label{lsjdhg}
Vector fields $T$ and $S$ on $\MU$ satisfy equalities
\begin{equation}                                                  \label{abcnft}
  (T,S)=0,\qquad T^2+S^2=0,
\end{equation}
if and only if vector field $S$ is given by Eq.(\ref{abnxhh}) with arbitrary
functions $B>0$ and $\psi\in\MR$, and the second vector field has the form
\begin{equation}                                                  \label{anbsgt}
  T=-\frac B{\sqrt{|h|}}\big[\cosh\psi X^{\prime2}+\sinh\psi(\dot X,X')\big]\dot X
  +\frac B{\sqrt{|h|}}\big[\sinh\psi\dot X^2+\cosh\psi(\dot X,X')\big]X'.
\end{equation}
\end{lemma}
\begin{proof}
Substitution of Eqs.~(\ref{abnxhh}) into the orthogonality condition
(\ref{abcnft}) yields equation
\begin{equation}                                                  \label{edsfwr}
  \frac{(T,S)}{\cosh\vf\cosh\psi}
  =\tanh\psi\dot X^2+(1+\tanh\vf\tanh\psi)(\dot X,X')
  +\tanh\vf X^{\prime2}=0
\end{equation}
which imply
\begin{equation*}
  \tanh\vf=-\frac{\tanh\psi\dot X^2+(\dot X,X')}{X^{\prime2}
  +\tanh\psi(\dot X,X')}.
\end{equation*}
Then vector field $T$ can be written in the form
\begin{equation*}
  T=-\tilde A\big[X^{\prime2}+\tanh\psi(\dot X,X')\big]\dot X
  +\tilde A\big[\tanh\psi\dot X^2+(\dot X,X')\big]X',
\end{equation*}
where
\begin{equation*}
  A:=\tilde A\sqrt{\big[\tanh\psi\dot X^2+(\dot X,X')\big]^2
  +\big[X^{\prime2}+\tanh\psi(\dot X,X')\big]^2}.
\end{equation*}
Now algebraic equation $T^2+S^2=0$ has the unique solution
\begin{equation*}
  \tilde A=\frac B{\sqrt{|h|}}\cosh\psi.
\end{equation*}
As a result, we obtain solution (\ref{anbsgt}) for arbitrary $S$.
\end{proof}

The similar statement can be formulated considering the vector field $T$ as
independent variable. To this end we have to solve Eq.(\ref{edsfwr}) with
respect to $\psi$ and afterwards find $S$.

Vector fields $T$ and $S$ are defined on the whole string worldsheet and lie
in the tangent space to the Minkowskian space $\MT(\MR^{1,D-1})$. The
differential map of the embedding $\MU\hookrightarrow\MR^{1,D-1}$ acts on
vectors as follows
\begin{equation*}
  \MT(\MU)\ni\quad t=t^\al\pl_\al,~s=s^\al\pl_\al~\mapsto~
  T=t^\al\pl_\al X^\Sa\pl_\Sa,~S=s^\al\pl_\al X^\Sa\pl_\Sa\quad
  \in\MT(\MR^{1,D-1}),
\end{equation*}
where $t$ and $s$ are vector fields on $\MU$. Note the relations
\begin{equation}                                                  \label{abbcnd}
  T^2=t^2\qquad S^2=s^2,\qquad (T,S)=(t,s),
\end{equation}
which follow from the definition of the induced metric. Comparing the above
formulae with Eqs.(\ref{abnxhh}) and (\ref{anbsgt}) allows us to define
vector fields on $\MU$:
\begin{equation}                                                  \label{abcndg}
\begin{split}
  t=&-\frac B{\sqrt{|h|}}\big[\cosh\psi X^{\prime2}+\sinh\psi(\dot X,X')\big]\pl_0
  +\frac B{\sqrt{|h|}}\big[\sinh\psi\dot X^2+\cosh\psi(\dot X,X')\big]\pl_1.
\\
  s=&~~B\sinh\psi\pl_0+B\cosh\psi\pl_1.
\end{split}
\end{equation}
These vectors can be easily rewritten in the form
\begin{equation}                                                  \label{avsfrd}
  t=\ve^{\al\bt}s_\bt\pl_\al,\qquad s=s^\al\pl_\al,
\end{equation}
where $\ve^{\al\bt}$ is the totally antisymmetric second rank tensor,
$\ve^{01}=-1/\sqrt{|h|}$, and components $s^\al$ are arbitrary. It immediately
implies equalities (\ref{axncjh}). There is a one-to-one correspondence between
vector components (\ref{avsfrd})
\begin{equation}                                                  \label{abcndj}
  t^\al=\ve^{\al\bt}s_\bt\quad\Leftrightarrow\quad s_\al=\ve_{\al\bt}t^\bt.
\end{equation}
That is we can take either $t^\al$ or $s^\al$ as independent variables. Now we
have to find the condition of their commutativity.
\begin{lemma}                                                     \label{lkwioi}
Vector fields $t$ and $s$ on $\MU$ related by equalities (\ref{abcndj}) commute
if and only if
\begin{equation}                                                  \label{amvnfu}
  t_\al=\frac{\pl_\al\chi}{\pl\chi^2},\qquad
  \pl\chi^2:=h^{\al\bt}\pl_\al\chi\pl_\bt\chi>0,
\end{equation}
where $\chi$ is a nontrivial solution of the wave equation
\begin{equation}                                                  \label{amnfht}
  \Box_{(h)}\chi:=h^{\al\bt}\nb_\al\nb_\bt\chi=0,
\end{equation}
satisfying $\pl\chi^2>0$.

For any nontrivial solution of Eq.(\ref{amnfht}) satisfying the condition
$\pl\chi^2>0$, vector fields (\ref{abcndj}) and (\ref{amvnfu}) commute.
\end{lemma}
Note that for different nontrivial solutions of the wave equation (\ref{amnfht})
the pairs of vector fields $t$, $s$ differ in general.
\begin{proof}
The equality $[t,s]=0$ together with conditions (\ref{abcndj}) is equivalent
to the system of equations
\begin{equation}                                                  \label{awjjuy}
  s_\al\nb_\bt s^\bt-\frac12\nb_\al s^2-s^\bt\nb_\bt s_\al=0.
\end{equation}
This equation is equivalent to the same equation for $t$:
\begin{equation}                                                  \label{alvnfu}
  t_\al\nb_\bt t^\bt-\frac12\nb_\al t^2-t^\bt\nb_\bt t_\al=0.
\end{equation}
Contraction of Eq.(\ref{awjjuy}) with $s^\al$ yields
\begin{equation}                                                  \label{ansjyg}
  \nb_\al\left(\frac{s^\al}{s^2}\right)=0\qquad\Leftrightarrow\qquad
  \ve^{\al\bt}\nb_\al\left(\frac{t_\bt}{t^2}\right)=0.
\end{equation}
Since Christoffel's symbols are symmetric, the covariant derivatives in the
last equation can be replaced by the partial ones. Then, due to the Poincar\'e
lemma, there exists such function $\chi$ that
\begin{equation*}
  \frac{t_\al}{t^2}=\pl_\al\chi\qquad\Rightarrow\qquad t^2=\frac1{\pl\chi^2}
\end{equation*}
on arbitrary connected and simply connected domain $\MU$, in particular, on the
whole plane $\MU=\MR^2$.

We require $\pl\chi^2>0$ because vector field $t$ on $\MU$ must be timelike,
$t^2>0$. Consequently, representation (\ref{amvnfu}) is valid for components
$t_\al$. Substitution of equality (\ref{amvnfu}) into Eq.(\ref{alvnfu}) yields
\begin{equation}                                                  \label{abcvdf}
  \frac{\pl_\al\chi\,\square_{(h)}\chi}{\pl\chi^2}=0.
\end{equation}
Since $(\pl_\al\chi)\ne0$ and $\pl\chi^2>0$, we obtain Eq.(\ref{amnfht}) for the
unknown function $\chi$.

It is clear that for any nontrivial solution of the wave equation (\ref{amnfht})
the vector fields $t$ and $s$ exist, commute, and have properties
(\ref{axncjh}).
\end{proof}

Thus commuting vector fields $t$ and $s$ with properties (\ref{axncjh}) have
generally the following form
\begin{equation}                                                  \label{andytr}
  t=\frac{h^{\al\bt}\pl_\bt\chi}{\pl\chi^2}\pl_\al,\qquad
  s=\frac{\ve^{\al\bt}\pl_\bt\chi}{\pl\chi^2}\pl_\al,
\end{equation}
where $\chi$ is an arbitrary solution of the wave equation (\ref{amnfht}) such
that $\pl\chi^2>0$. That is we have described the total arbitrariness existing
in vector fields if Eq.(\ref{amnfht}) has many solutions.

\begin{com}
We did not use in Lemma \ref{lkwioi} the fact that vector fields $t$ and $s$
were obtained by the embedding $\MU\hookrightarrow\MR^{1,D-1}$. It is sufficient
to consider two vector fields related by Eq.~(\ref{avsfrd}). The properties
(\ref{axncjh}) are easily verified without embedding.
\qed\end{com}

Suppose that the determinant of the induced metric $h_{\al\bt}$ is nonzero on
the whole plane $(\tau,\s)\in\MR^2$ and separated from $0$ and $\pm\infty$ at
infinity:
\begin{equation}                                                  \label{abnshg}
  0<\e\le\underset{\tau^2+\s^2\to\infty}\lim|\det h_{\al\bt}|\le M<\infty,
\end{equation}
where $\e$ and $M$ are some constants. It is well known that the Cauchy problem
for the hyperbolic equation (\ref{amnfht}) has unique solution $\chi$ on the
whole plain, if the Cauchy data are given on a spacelike curve, say, $\tau=0$
(see, e.g.~\cite{Hadama32}, book IV, ch.\ I).
It is easily verified that there exist such Cauchy data that the inequality
$\pl\chi^2>0$ holds everywhere. This implies that nontrivial solution of the
wave equation  (\ref{amnfht}) exists on the whole plane $\MR^2$. There are many
such solutions, and they are parameterized by the Cauchy data.

Thus the vector fields $s$ and $t$ are given on the whole plane $\MR^2$. The
inequality (\ref{abnshg}) implies that component $t^0$ is separated from
zero and bounded on the plane including infinity. Thus Eqs.(\ref{ancjdy}) imply
\begin{equation*}
  \frac{\pl\tau}{\pl{\tilde\tau}}=t^0\qquad\Rightarrow\qquad
  \tilde\tau\sim\int^\infty\frac{d\tau}{t^0}.
\end{equation*}
The last integral is divergent and thus the coordinate $\tilde\tau$ runs over
the whole real line $\MR$. Similar statement is valid for the space coordinate
$\tilde\s$. Consequently, new coordinates cover the whole plane
$(\tilde\tau,\tilde\s)\in\MR^2$.

Thus we proved
\begin{theorem}                                                   \label{tdgefr}
Let an arbitrary metric $h_{\al\bt}$ of Lorentzian signature be given on the
whole plane $\MR^2$. Let it be nondegenerate at infinity (\ref{abnshg}). Then
there exists a surjective diffeomorphism on the plane
\begin{equation}                                                  \label{ehgdrv}
  \MR^2\ni\quad (x^\al)\mapsto\big(\tilde x^\al(x)\big)\quad\in\MR^2
\end{equation}
such that metric $h_{\al\bt}$ in new coordinate system has conformally flat form
\begin{equation}                                                  \label{ehhhgd}
  \tilde h_{\al\bt}:=h_{\g\dl}\frac{\pl x^\g}{\pl\tilde x^\al}
  \frac{\pl x^\dl}{\pl\tilde x^\bt}=\ex^{2\phi}\eta_{\al\bt},
\end{equation}
where $\phi(\tilde x)$ is some function on $\MR^2$ separated from $\pm\infty$
at infinity $\tilde\tau^2+\tilde\s^2\to\infty$.
\end{theorem}

In contrast to the local theorem (see, e.g., \cite{Petrov61,Vladim71})
stating the existence of the conformal gauge only in some neighborhood of every
point, the above theorem is global in a sense that it provides the existing of
the conformal gauge for a Lorentzian metric given on the whole plane
$x\in\MR^2$.

If metric $h_{\al\bt}$ and scalar field $\chi$ are given, then the
diffeomorphism (\ref{ehgdrv}) is defined uniquely up to shifts of coordinates
$\tilde\tau$ and $\tilde\s$ \big(constants of integration of
Eqs.(\ref{ancjdy})\big).

\begin{cor}
Let
\begin{equation}                                                  \label{anvbft}
  \widetilde\MU_0:=\lbrace (\tilde\tau,\tilde\s)\in\MR^2:\quad
  \tilde\s\in[\tilde\s_1,\tilde\s_2],~\tilde\tau\in\MR\rbrace
\end{equation}
be closed vertical strip with straight boundaries on the plane of new
coordinates $\tilde\tau,\tilde\s$ and assumptions of theorem \ref{tdgefr} hold.
Then there exists diffeomorphism (\ref{ehgdrv}) of a closed domain
$(\tau,\s)\in\overline\MU\subset\MR^2$ bounded by integral curves
$x(\tilde\tau,\tilde\s_{1,2})$:
\begin{equation*}
  \frac{\pl x(\tilde\tau,\tilde\s_{1,2})}{\pl\tilde\tau}=t_{1,2},
\end{equation*}
where $t_{1,2}$ are inverse images of vector fields $\pl/\pl\tilde\tau$ on the
boundaries of $\widetilde\MU_0$.
\qed\end{cor}
\begin{proof}
Internal and boundary points are mapped under diffeomorphism into internal and
boundary points, respectively.
\end{proof}

To clarify the arbitrariness in coordinates $\tilde\tau$, $\tilde\s$ defined by
the function $\chi$ we consider

{\bf Example 1}. Let the induced metric be conformally flat:
\begin{equation}                                                  \label{abdnft}
  h_{\al\bt}dx^\al dx^\bt=\ex^{2\phi}(d\tau^2-d\s^2)
  =\ex^{2\phi}d\xi d\eta,\qquad\phi=\phi(x),
\end{equation}
where light cone coordinates $\xi:=\tau+\s$, $\eta:=\tau-\s$ are introduced.
Then the wave equation (\ref{amnfht}) is reduced to the flat d'Alembert equation
\begin{equation*}
  \Box_{(h)}\chi=(\pl^2_0-\pl^2_1)\chi=0.
\end{equation*}
Its general solution is given by two arbitrary sufficiently smooth functions
\begin{equation*}
  \chi=F(\xi)+G(\eta).
\end{equation*}
We choose only the functions which satisfy inequality
\begin{equation*}
  \pl\chi^2=4\ex^{-2\phi}F'G'>0\qquad\Rightarrow\qquad F'G'>0,
\end{equation*}
where prime denotes differentiation by the corresponding argument. Then the
metric takes the form
\begin{equation*}
  \tilde h_{\al\bt}d\tilde x^\al d\tilde x^\bt
  =\frac{\ex^{2\phi}}{4F'G'}d\tilde\xi d\tilde\eta
\end{equation*}
in new coordinates $\tilde\xi:=\tilde\tau+\tilde\s$,
$\tilde\eta:=\tilde\tau-\tilde\s$. This metric corresponds to the conformal
transformation
\begin{equation*}
  \tilde\xi:=2F(\xi),\qquad\tilde\eta:=2G(\eta).
\end{equation*}
Thus the arbitrariness in definition of the vector fields described in Lemma
\ref{lkwioi} corresponds to conformal transformations on the string worldsheet.
\qed

Let us find the image of vector fields $t$ and $s$ (\ref{andytr}) under the
coordinate transformation $x^\al\mapsto\tilde x^\al(x)$ where
$(\tilde x^\al)=(\tilde\tau,\tilde\s)$. The definition of new coordinates
(\ref{ancjdy}) implies the expression for the inverse Jacobi matrix
\begin{equation}                                                  \label{ancbdi}
  J^{-1}{}_\al{}^\bt=\frac{\pl x^\bt}{\pl \tilde x^\al}
  =\begin{pmatrix} t^0 & t^1 \\ s^0 & s^1 \end{pmatrix}.
\end{equation}
This expression yields the Jacobian of the coordinate transformation
\begin{equation*}
  J:=\det J_\al{}^\bt=(t^0s^1-t^1s^0)^{-1}=\sqrt{|h|}\pl\chi^2,
\end{equation*}
where representation (\ref{andytr}) is used, and the Jacobi matrix
\begin{equation}                                                  \label{andmui}
  J_\al{}^\bt:=\frac{\pl\tilde x^\bt}{\pl x^\al}
  =\sqrt{|h|}\pl\chi^2\begin{pmatrix} s^1 & -t^1 \\ -s^0 & t^0 \end{pmatrix}.
\end{equation}
Vector fields are transformed under the differential of the map:
\begin{equation*}
  t^\al\pl_\al\mapsto \tilde t^\al\tilde\pl_\al
  :=t^\bt\pl_\bt\tilde x^\al\tilde\pl_\al=\pl_{\tilde\tau},\qquad
  s^\al\pl_\al\mapsto \tilde s^\al\tilde\pl_\al
  :=s^\bt\pl_\bt\tilde x^\al\tilde\pl_\al=\pl_{\tilde\s}.
\end{equation*}
It means that integral curves of the vector fields $t$ and $s$ are perpendicular
straight lines on the plane $(\tilde\tau,\tilde\s)\in\MR^2$.

In the target Minkowskian space $\MR^{1,D-1}$ the vector fields $t$ and $s$
have the form
\begin{equation}                                                  \label{abvxfd}
\begin{split}
  T=&\frac{h^{\al\bt}\pl_\bt\chi}{\pl\chi^2}\pl_\al X
  =\frac{h^{0\bt}\pl_\bt\chi}{\pl\chi^2}\dot X
  +\frac{h^{1\bt}\pl_\bt\chi}{\pl\chi^2}X',
\\
  S=&\frac{\ve^{\al\bt}\pl_\bt\chi}{\pl\chi^2}\pl_\al X
  =\frac{\ve^{0\bt}\pl_\bt\chi}{\pl\chi^2}\dot X
  +\frac{\ve^{1\bt}\pl_\bt\chi}{\pl\chi^2}X'.
\end{split}
\end{equation}

It was already mentioned, that the wave equation (\ref{amnfht}) has nontrivial
solutions $\chi$, $\pl\chi^2\ne0$, on the whole plane $(\tau,\s)\in\MR^2$. The
wave equation has the same form in new coordinates $\tilde\tau,\tilde\s$ as on
the Minkowskian plane. Therefore a general solution of the wave equation
(\ref{amnfht}) is
\begin{equation}                                                  \label{asvdbg}
  \chi=F(\tilde\xi)+G(\tilde\eta)=F(\tilde\tau+\tilde\s)+G(\tilde\tau-\tilde\s),
  \qquad(\tilde\xi,\tilde\eta)\in\MR^2,
\end{equation}
where $F$ and $G$ are two arbitrary sufficiently smooth functions of one
argument. It implies equalities:
\begin{equation*}
\begin{split}
  \pl_0\chi:=\frac{\pl\chi}{\pl\tau}=&\sqrt{|h|}\pl\chi^2\left[(s^1-t^1)F'
  +(s^1+t^1)G'\right],
\\
  \pl_1\chi:=\frac{\pl\chi}{\pl\s}=&\sqrt{|h|}\pl\chi^2\left[(t^0-s^0)F'
  +(t^0+s^0)G'\right],
\end{split}
\end{equation*}
were Jacobi's matrices (\ref{andmui}) are used. Then the components of vectors
$\hat t$ and $\hat s$ corresponding to solution (\ref{asvdbg}) are
\begin{equation}                                                  \label{avbcgd}
\begin{split}
  \hat t^0=&-\frac1{\sqrt{|h|}}\left[X^{\prime2}(s^1-t^1)F'
  +X^{\prime2}(s^1+t^1)G'+(\dot X,X')(s^0-t^0)F'+(\dot X,X')(s^0+t^0)G'\right],
\\
  \hat t^1=&~~\frac1{\sqrt{|h|}}\left[(\dot X,X')(s^1-t^1)F'
  +(\dot X,X')(s^1+t^1)G'+\dot X^2(s^0-t^0)F'+\dot X^2(s^0+t^0)G'\right],
\\
  \hat s^0=&~~(s^0-t^0)F'+(s^0+t^0)G',
\\
  \hat s^1=&~~(s^1-t^1)F'+(s^1+t^1)G'.
\end{split}
\end{equation}
The components on the left hand side are marked by the hat because vectors
$\hat t$ and $\hat s$ differ in general from the ones which were used for the
transformation of the wave equation to the flat form. However for some functions
$F$ and $G$ we must get identities.
\begin{prop}                                                      \label{pkfjuy}
Equalities
\begin{equation*}
  \hat t=t,\qquad\hat s=s
\end{equation*}
become identities if and only if arbitrary functions in solution (\ref{asvdbg})
are linear: $F'=G'=1/2$.
\end{prop}
\begin{proof}
Straightforward calculations.
\end{proof}
Thus for
\begin{equation}                                                  \label{amsdhg}
  \chi=\frac12\tilde\xi+\frac12\tilde\eta=\tilde\tau,
\end{equation}
equalities (\ref{andytr}) become identities. The vector fields
$\pl_{\tilde\tau}$ and $\pl_{\tilde\s}$ are obtained from the fields $t$ and $s$
by the differential of the map $(\tau,\s)\mapsto(\tilde\tau,\tilde\s)$ only for
this solution of the wave equation.

Thus to find the diffeomorphism (\ref{ehgdrv}) in explicit form for a given
metric $h_{\al\bt}$, we have to (i) find a nontrivial solution of the wave
equation (\ref{amnfht}), (ii) construct the vector fields $t$ and $s$ using
Eqs.(\ref{amvnfu}), (\ref{abcndj}), and (iii) find a general solution of the
system of equations (\ref{ancjdy}). We have proved that this problem does have
many solutions (the whole arbitrariness is contained in the choice of nontrivial
solution of the wave equation).
%******************************************************************************
\section{Open string}
%*******************************************************************************
Now we consider an open string whose worldsheet $\overline\MU$ is an infinite
strip on the plane $(\tau,\s)\in\MR^2$ with two, probably, curved boundaries:
the left $\g_\Sl$ and right $\g_\Sr$ boundaries. The induced metric on the
boundaries is degenerate, and results of the previous section must be revised.
First, we assume that metric is not degenerate and return to this problem later.

If the metric is nondegenerate on $\overline\MU$ including the boundaries, then
we continue it on the whole plain in some sufficiently smooth manner. As the
consequence of theorem \ref{tdgefr} there is a diffeomorphism (\ref{ehgdrv})
after which the metric becomes conformally flat. The problem is that the
boundaries $\g_{\Sl,\Sr}$ on the plain $\tilde\tau,\tilde\s$ may be not
straight vertical lines. However there are residual diffeomorphisms in the
form of conformal maps of $\tilde\tau,\tilde\s$ coordinates. We now show that it
is enough to straighten the strip.

Remember that we do not consider shifts of the plain
$\tilde\xi,\tilde\eta\in\MR^2$ as a whole which preserve the conformal form of
the metric but is not a conformal map.

Let boundary equations after diffeomorphism
$(\tau,\s)\mapsto(\tilde\tau,\tilde\s)$ be (see Fig.~\ref{confmapsst})
 \begin{equation}                                                 \label{eaffqj}
   \g_\Sl:\quad\tilde\eta=\tilde\eta_\Sl(\tilde\xi),\qquad
   \g_\Sr:\quad\tilde\eta=\tilde\eta_\Sr(\tilde\xi),\qquad
   \tilde\xi\in\MR,
\end{equation}
where functions $\tilde\eta_{\Sl,\Sr}\in\CC^1(\MR)$ have properties:
\begin{equation*}
  \tilde\eta_\Sl>\tilde\eta_\Sr,\qquad
  0<\e\le\frac{d\tilde\eta_{\Sl,\Sr}}{d\tilde\xi}\le M<\infty,\qquad\e,M\in\MR
\end{equation*}
for all $\tilde\xi\in\MR$ including infinite points.
\begin{figure}[hbt]%----------------------------------------------------------
\hfill\includegraphics[width=\textwidth]{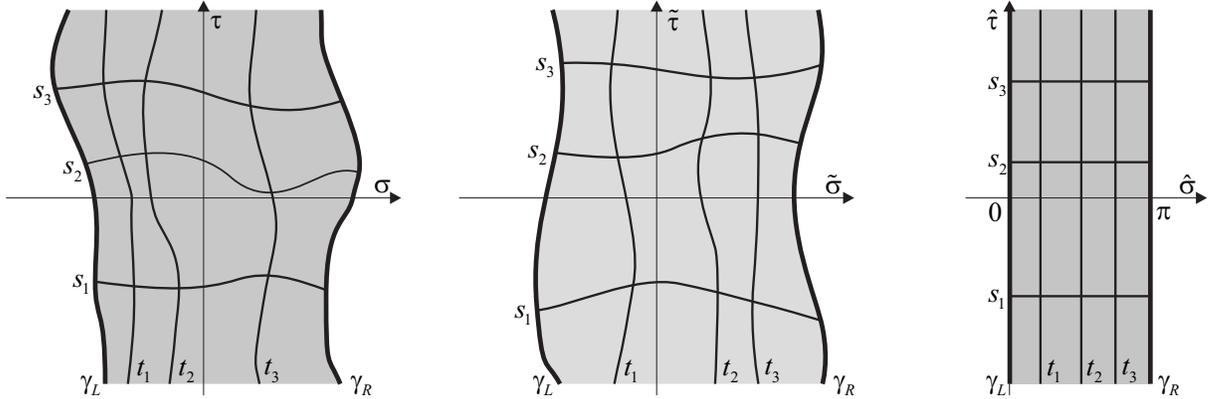}
\hfill {}
\centering\caption{Open string worldsheet in different coordinates $(\tau,\s)$,
$(\tilde\tau,\tilde\s)$, and $(\hat\tau,\hat\s)$. Three integral curves are
shown both for timelike $t$ and spacelike $s$ vector fields.}
\label{confmapsst}
\end{figure}%-------------------------------------------------------------------

\begin{theorem}                                                   \label{thhdgy}
The conformal transformation
\begin{equation*}
  \hat\xi=F(\tilde\xi),\qquad\hat\eta=G(\tilde\eta),\qquad F,G\in\CC^1(\MR),
\end{equation*}
such that the boundaries (\ref{eaffqj}) of an open string worldsheet become
straight vertical lines
\begin{equation}                                                  \label{ajhdgt}
   \g_\Sl:\quad\hat\eta=\hat\xi,\qquad
   \g_\Sr:\quad\hat\eta=\hat\xi-2\pi,\qquad\hat\xi\in\MR
\end{equation}
on the plain $\hat\xi,\hat\eta\in\MR^2$ exists.
\end{theorem}
\begin{proof}
Let us straighten first the left boundary by conformal transformation
\begin{equation*}
  \check\xi=F(\tilde\xi),\qquad\check\eta=G(\tilde\eta).
\end{equation*}
It is necessary and sufficient to fulfill the condition
\begin{equation*}
 F(\tilde\xi)=G\big(\tilde\eta_\Sl(\tilde\xi)\big),\qquad\forall\tilde\xi\in\MR,
\end{equation*}
in order that the left boundary to be vertical straight line going through the
origin. This equation uniquely defines the function $G$ for a given $F$
because the function $\tilde\eta_\Sl$ is strictly monotonic, the function $F$
being arbitrary.

Now we consider the right boundary. After straightening the left boundary, we
are left with the coordinate transformation
\begin{equation*}
  \hat\xi=F(\check\xi),\qquad\hat\eta=F(\check\eta),
\end{equation*}
which does not change the left boundary and is described by one arbitrary
function $F$. It is necessary and sufficient to satisfy the functional equation
\begin{equation}                                                  \label{abcvfr}
  F\big(\check\eta_\Sr(\check\xi)\big)=F(\check\xi)-2\pi,\qquad
  \forall\check\xi\in\MR,
\end{equation}
after which the right boundary becomes the vertical straight line going through
the point $(\hat\tau,\hat\s)=(0,\pi)$. This functional equation for $F$ has many
solutions. Indeed, the map
\begin{equation*}
  f:\quad\MR\ni\qquad\check\xi\mapsto\check\eta_\Sr(\check\xi)\qquad\in\MR
\end{equation*}
is a bijective map of real lines, and $\check\eta_\Sr(\check\xi)<\check\xi$ for
all $\check\xi\in\MR$. Therefore the sufficiently smooth cyclic group
$\lbrace f^k,~k\in\MZ\rbrace$ is defined. Consequently, the function $F$ with
property $F'>0$ can be arbitrary defined on the fundamental domain, say,
$[\check\eta_\Sr(0),0]$ and then extended on the whole real line using
Eq.(\ref{abcvfr}). If the function satisfies Eq.(\ref{abcvfr}) at the ends of
the fundamental domain then we obtain continuous function on $\MR$ but its
derivative may be discontinuous. The function $F$ must be $\CC^1$ in order to
define the conformal transformation. To avoid possible discontinuities in
derivatives, we differentiate Eq.(\ref{abcvfr}):
\begin{equation*}
  F'\big(\check\eta_\Sr(\check\xi)\big)\frac{d\check\eta_\Sr}{d\check\xi}
  =F'(\check\xi).
\end{equation*}
Since $d\check\eta/d\check\xi\ge\e>0$, then we define arbitrary the derivative
$F'>0$ on the fundamental domain $[\check\eta_\Sr(0),0]$ such that equation
\begin{equation*}
  F'\big(\check\eta_\Sr(0)\big)\left.\frac{d\check\eta_\Sr}{d\check\xi}
  \right|_{\check\eta_\Sr(0)}=F'(0)
\end{equation*}
holds at the ends, and continue it on the real line. Then $F$ is the primitive
of $F'$ with the constant of integration defined by Eq.(\ref{abcvfr}).
\end{proof}

So, if the metric is not degenerate on the boundaries of an open string
worldsheet, then there exists such global $\CC^1$ coordinate transformation
that the transformed metric is conformally flat (\ref{ehhhgd}) on the whole
vertical strip with straight boundaries $\tilde\s=0$ and $\tilde\s=\pi$. This
statement follows from theorems \ref{tdgefr} and \ref{thhdgy} because the
conformal transformation is a diffeomorphism, and diffeomorphisms form a group.

Now we discuss an open Nambu--Goto string for which the induced metric on the
boundaries is degenerate due to boundary conditions. Let us parameterize metric
$h_{\al\bt}$ by its determinant $-\varrho^4$ and ``unimodular metric''
$k_{\al\bt}$:
\begin{equation}                                                  \label{abcndf}
  h_{\al\bt}:=\varrho^2k_{\al\bt},\qquad \det k_{\al\bt}:=-1,\qquad\varrho\ge0,
\end{equation}
separating its determinant explicitly. For nondegenerate metric $h_{\al\bt}$,
the inverse transformation is
\begin{equation*}
  \varrho=|\det h_{\al\bt}|^{1/4},\qquad k_{\al\bt}=\varrho^{-2}h_{\al\bt},
  \qquad\varrho>0.
\end{equation*}
Definition (\ref{abcndf}) implies that the variable $\varrho$ is the scalar
density of degree $-1/2$ and unimodular metric is the second rank tensor density
of degree 1.

Note that the unimodular metric $k_{\al\bt}$ is additionally multiplied by the
Jacobian under arbitrary transformation of coordinates because it is a tensor
density of degree 1. It means that the induced and unimodular metrics take the
conformally flat form simultaneously.

The boundary condition (\ref{ubvbxg}) has the form
\begin{equation*}
  n^0\dot X^\Sa+n^1X^{\prime\Sa}=0\qquad\Rightarrow\qquad
  X^{\prime\Sa}=-\frac{n^0}{n^1}\dot X^\Sa,
\end{equation*}
because the normal vector to the boundaries $(n^\al)=(n^0,n^1)$ must be
spacelike and consequently $n^1\ne0$. As the consequence, the metric degenerates
on the boundaries:
\begin{equation*}
  \det h_{\al\bt}=-\rho^4=\dot X^2 X^{\prime2}-(\dot X,X')^2\to0.
\end{equation*}
Therefore vector fields $t$ and $s$ (\ref{axncjh}) become null
\begin{equation*}
  t^2=\rho^2k_{\al\bt}t^\al t^\bt\to0,\qquad
  s^2=\rho^2k_{\al\bt}s^\al s^\bt\to0
\end{equation*}
at the ends of the string with respect to the induced metric but not with
respect to the unimodular metric, as we shall see later.

Now we construct new vector fields for the unimodular metric $k_{\al\bt}$
satisfying relations
\begin{equation}                                                  \label{avsbsf}
  (t,s):=k_{\al\bt}t^\al s^\bt=0,\qquad t^2+s^2=0,
\end{equation}
where
\begin{equation*}
  t^2:=k_{\al\bt}t^\al t^\bt\qquad s^2:=k_{\al\bt}s^\al s^\bt.
\end{equation*}
Equalities (\ref{avsbsf}) are equivalent to original equations (\ref{axncjh}) in
internal points of $\MU$ and extended on boundaries $\pl\MU$ by continuity.
It is clear that properties (\ref{avsbsf}) do not contradict Eqs.(\ref{axncjh}).
Now we prove that new vector fields $t$ and $s$ exist and coincide with the
original ones for $h_{\al\bt}$.

Formulae (\ref{ehsdhg}) for metric components after the coordinate
transformation have the same form. In addition,
\begin{equation*}
  \tilde h_{00}=-\tilde h_{11}=\varrho^2k_{\al\bt}t^\al t^\bt,\qquad
  \tilde h_{01}=0.
\end{equation*}

Equation (\ref{abcvdf}) in new variables does not depend on $\rho$ and
therefore function $\chi$ must satisfy equation
\begin{equation}                                                  \label{ennngh}
  \square_{(k)}\chi:=\pl_\al\big(k^{\al\bt}\pl_\bt\chi\big)=0.
\end{equation}
This wave equation has many solutions on the whole plain $\MR^2$ because depends
on nondegenerate unimodular metric $k_{\al\bt}$. It implies that vector fields
$t$ and $s$ exist and do not depend on $\varrho$:
\begin{equation}                                                  \label{eersew}
  t^\al=\frac{k^{\al\bt}\pl_\bt\chi}{k^{\g\dl}\pl_\g\chi\pl_\dl\chi}\qquad
  s^\al=\frac{\hat\ve^{\al\bt}\pl_\bt\chi}{k^{\g\dl}\pl_\g\chi\pl_\dl\chi},
\end{equation}
where
\begin{equation*}
  \hat\ve^{\al\bt}=\rho^2\ve^{\al\bt}=
  \begin{pmatrix} 0 & -1 \\ 1 & ~~0 \end{pmatrix}
\end{equation*}
is the totally antisymmetric tensor density of degree $-1$.

One can easily verify that in new variables the Euler--Lagrange equations for
bosonic string (\ref{uvbsju}) take the form
\begin{equation}                                                  \label{avdbcf}
  \sqrt{|h|}\square_{(h)} X^\Sa=\pl_\al\big(\sqrt{|h|}h^{\al\bt}
  \pl_\bt X^\Sa\big)=\pl_\al\big(k^{\al\bt}\pl_\bt X^\Sa\big)=0,
\end{equation}
i.e.\ do not depend on $\varrho$. They must be solved with the boundary
condition
\begin{equation}                                                  \label{avdbck}
  n^\al\pl_\al X^\Sa\big|_{\g_{\Sl,\Sr}}=0,
\end{equation}
which does not depend on $\rho$ too.

Consequently, the problem is reduced to solution of Eqs.(\ref{avdbcf}) with
boundary conditions (\ref{avdbck}) for an open Nambu--Goto string. To make the
transformation of coordinates $(\tau,\s)\mapsto(\tilde\tau,\tilde\s)$ we have to
find the unimodular metric $k_{\al\bt}$ for a given metric $h_{\al\bt}$, choose
a solution $\chi$ of the wave equation (\ref{ennngh}) satisfying the condition
$k^{\g\dl}\pl_\g\chi\pl_\dl\chi>0$, construct the vector fields $t$ and $s$
using formulae (\ref{eersew}), and, finally, integrate Eqs.(\ref{ancjdy}).
Therefore the corollary of theorem \ref{tdgefr} is valid also for an open
string. If needed, after solution of this problem, one can compute the
conformal factor for the induced metric (\ref{ehhhgd}) using equation
\begin{equation}                                                  \label{abcvdj}
  \ex^{2\phi}=\varrho^2k_{\al\bt}t^\al t^\bt.
\end{equation}
Sure, it is zero on the boundaries because $\rho\to0$. Thus we proved the
existence of the global conformal gauge for an open string.
%******************************************************************************
\section{Closed string}
%*******************************************************************************
In the initial coordinates $\tau,\s\in\MR^2$, the fundamental domain of a closed
string worldsheet
is given by an infinite strip with timelike boundaries which are identified. The
identification can be performed in many ways and therefore requires definition.
Here we describe the method adopted in string theory.

We showed in the previous section that there is the global diffeomorphism
$(\tau,\s)\mapsto(\hat\tau,\hat\s)$ which maps an arbitrary infinite strip with
timelike boundaries on the vertical strip with straight boundaries where metric
becomes conformally flat. The same procedure can be performed for the
fundamental domain of a closed string. Without loss of generality, we assume
that boundaries go through points $\hat\s=\pm\pi$, as is usually supposed in
string theory. Then the boundary identification is written as the periodicity
condition for every value of the timelike coordinate $\hat\tau$:
\begin{equation}                                                  \label{abbcvf}
  \left.\frac{\pl^k X^\Sa}{\pl\hat\s^k}\right|_{\hat\s=-\pi}=
  \left.\frac{\pl^k X^\Sa}{\pl\hat\s^k}\right|_{\hat\s=\pi},\qquad
  \forall\Sa,\quad\forall\hat\tau,\quad k=0,1,2,\dotsc.
\end{equation}
That is we continuously glue the coordinate functions themselves and their
derivatives up to the needed order. In the initial coordinate system this
condition is written in the covariant form
\begin{equation}                                                  \label{ancbgf}
  \nb_s^k X^\Sa\big|_{\hat\s=-\pi}=\nb_s^k X^\Sa\big|_{\hat\s=\pi},
\end{equation}
where $\nb_s:=s^\al\nb_\al$ is the covariant derivative for the Levi--Civita
connection along the vector field $s$ which is the pullback of the vector field
$\pl/\pl\hat\s$ under the diffeomorphism $(\tau,\s)\mapsto(\hat\tau,\hat\s)$.
The properties of vector fields $t,s$ (\ref{axncjh}) imply that the
covariant derivatives are taken along normal vectors to the boundaries.
However it is not clear at the beginning for which value of $\tau$ on the left
and right the identification takes place, because we have to find the
diffeomorphism $(\tau,\s)\mapsto(\hat\tau,\hat\s)$ explicitly. Simply speaking
we firstly transform the metric to conformally flat form and afterwards perform
the natural gluing.
%******************************************************************************
\section{The Euclidean signature}
%*******************************************************************************
Many of the preceeding statements do not depend on the signature of the metric,
and the results can be generalized for an arbitrary Riemannian positive definite
metric. In this section, we suppose that sufficiently smooth two-dimensional
Riemannian metric $h_{\al\bt}$, $\det h_{\al\bt}>0$, is given on the whole plane
$\MR^2$. It is not necessarily induced by some embedding.

Let us consider two vector fields $t=t^\al\pl_\al$ and $s=s^\al\pl_\al$,
related by equation
\begin{equation}                                                  \label{anbsgf}
  t^\al=\ve^{\al\bt}s_\bt,\qquad (t^\al)\ne0,
\end{equation}
where $\ve^{\al\bt}:=h^{\al\g}h^{\bt\dl}\ve_{\g\dl}$ is the totally
antisymmetric second rank tensor, $\ve_{12}:=\sqrt{\det h_{\al\bt}}$. They have
the properties
\begin{equation}                                                  \label{abbnsd}
  t^2-s^2=0,\qquad (t,s)=0,
\end{equation}
where the scalar product is defined by $h_{\al\bt}$. The changing of the metric
signature results in changing of one sign in the first equation in
Eqs.~(\ref{abbnsd}) as compared to Eq.~(\ref{axncjh}).

After the coordinate transformation $\tau,\s\mapsto\tilde\tau,\tilde\s$ defined
by Eqs.~(\ref{ancjdy}) the metric becomes conformally flat:
\begin{equation*}
  ds^2=h_{\al\bt}\frac{\pl x^\al}{\pl\tilde x^\g}
  \frac{\pl x^\bt}{\pl\tilde x^\dl}d\tilde x^\g d\tilde x^\dl
  =t^2d\tilde\tau^2+s^2d\tilde\s^2=t^2(d\tilde\tau+d\tilde\s^2),
\end{equation*}
where $t^2:=h_{\al\bt}t^\al t^\bt$.
\begin{lemma}
Vector fields $t$ and $s$ on $\MU$ related by Eq.~(\ref{anbsgf}) commute if and
only if
\begin{equation}                                                  \label{amvnft}
  t_\al=\frac{\pl_\al\chi}{\pl\chi^2},
\end{equation}
where $\chi$ is a nontrivial solution of the Laplace--Beltrami equation
\begin{equation}                                                  \label{amnfhp}
  \triangle_{(h)}\chi:=h^{\al\bt}\nb_\al\nb_\bt\chi=0.
\end{equation}

For any nontrivial solution of Eq.(\ref{amnfhp}) vector fields (\ref{anbsgf})
and (\ref{amvnft}) commute.
\end{lemma}
\begin{proof}
Repeats the proof of Lemma \ref{lkwioi} which does not depend on the signature
of the metric, but now we obtain the Laplace--Beltrami equation. It is well
known that Eq.~(\ref{amnfhp}) has many solutions on a plane (harmonic
functions), and there is no local extremum. Therefore the requirements
$(\pl_\al\chi)\ne0$ and $\pl\chi^2\ne0$ are fulfilled automatically for any
nontrivial (nonconstant) solution of Eq.~(\ref{amnfhp}).
\end{proof}
So, the most general commuting vector fields having properties (\ref{abbnsd})
have the same form (\ref{andytr}) as in the Lorentzian case. The only difference
is that now an arbitrary function $\chi$ satisfies the Laplace--Beltrami
equation (\ref{amnfhp}) instead of the wave equation (\ref{amnfht}).

To clarify the meaning of the harmonic function $\chi$ we consider the example.

{\bf Example 2.}
Let the initial metric be conformally flat
\begin{equation*}
  ds^2=\ex^{2\phi}(dx^2+dy^2)=\ex^{2\phi}dz d\bar z,
\end{equation*}
where $\phi(x,y)$ is a real valued function, we introduced complex coordinate
$z:=x+iy$, and the bar denotes complex conjugation. After the coordinate
transformation $x^\al\mapsto\tilde x^\al$ defined by function $\chi$ it is
\begin{equation*}
  ds^2=t^2(d\tilde x^2+d\tilde y^2)=t^2d\tilde zd\bar{\tilde z},
\end{equation*}
where $\tilde z:=\tilde x+i\tilde y$ and
\begin{equation*}
  t^2=h_{\al\bt}t^\al t^\bt=\frac1{\pl\chi^2}.
\end{equation*}
The Laplace--Beltrami equation for conformally flat metric reduces to the
Laplace equation
\begin{equation*}
  \pl_z\pl_{\bar z}\chi=0.
\end{equation*}
Its general real valued solution is
\begin{equation*}
  \chi=w(z)+\bar w(\bar z),
\end{equation*}
where $w(z)$ is an arbitrary holomorphic function. Therefore
\begin{equation*}
  \pl\chi^2:=h^{\al\bt}\pl_\al\chi\pl_\bt\chi=\ex^{-2\phi}
  \pl_z(w+\bar w)\pl_{\bar z}(w+\bar w)=\ex^{-2\phi}\pl_z w\pl_{\bar z}\bar w,
\end{equation*}
and the metric is
\begin{equation*}
  ds^2=\ex^{2\phi}\pl_w z\pl_{\bar w}\bar z\, dwd\bar w.
\end{equation*}
Thus the transformation of coordinates defined by the harmonic function $\chi$
coincides with the conformal transformation $z\mapsto w(z)$.
\qed

The existence of vector fields $t$ and $s$ provides sufficient conditions for
the existence of the conformal gauge on the whole Euclidean plane $\MR^2$ for
metrics separated from zero and infinity (\ref{abnshg}, and analog of theorem
\ref{tdgefr} holds.

The Euclidean version of string theory is used in the path integral formulation
of quantum string theory, which assumes summation over Riemannian surfaces of
different genera. The Riemannian surfaces cannot be covered by a single
coordinate chart, and therefore we cannot talk about the conformal gauge on the
whole Riemannian surface. The results of the present section guarantee the
existence of the conformal gauge on the whole coordinate chart which is
diffeomorphic to $\MR^2$. Previous theorems provide sufficient conditions for
the existence of the conformal gauge only in some sufficiently small
neighbourhood of each point of the manifold.
%******************************************************************************
\section{Conclusion}
%*******************************************************************************
It was assumed for many years that there exists the global conformal gauge in
string theory though this statement was proved only locally. In fact, almost all
results were obtained under validity of this assumption which turns out to be
true quite unexpectedly at least to the author. We proved the global existence
of the conformal gauge for infinite, open, and closed strings. The transition
from local to global statement is based
on the global existence of the solution of the Cauchy problem for a
two-dimensional hyperbolic equation with varying coefficients \cite{Hadama32}
and is far from being obvious.

As a byproduct, we proved global existence of the conformal gauge for a general
two-dimensional Lorentzian metric defined on the whole plane $\MR^2$ which is
not necessarily induced by an embedding and is well known locally for a long
time (see, e.g.\ \cite{Petrov61,Vladim71}).

The existence theorem is also proved for a Riemannian positive definite metric
defined on the whole Euclidean plane. It generalizes previous results providing
the existence of the conformal gauge in some sufficiently small neighbourhood
of each point.

%\bibliography{3dgrav,book,gravity,math,my,qft}
%\bibliographystyle{unsrt}
\end{document}